\documentclass[aps,pra,twocolumn,amsmath,amssymb,superscriptaddress]{revtex4-1}
\usepackage{bm}
\usepackage{graphicx}

\usepackage{amsfonts,epsfig}
\usepackage{latexsym}
\usepackage{epsfig}
\usepackage{amsfonts}
\usepackage{amsthm}
\usepackage{mathrsfs}
\usepackage{natbib}
\usepackage{color,verbatim,graphics}
\usepackage{psfrag}
\usepackage{xfrac}

\DeclareMathOperator*{\argmin}{argmin}

\newtheorem{lmm}{Lemma}
\newtheorem{thm}{Theorem}

\newcommand{\be}{\begin{eqnarray}}
\newcommand{\ee}{\end{eqnarray}}
\newcommand{\ba}{\begin{eqnarray}}
\newcommand{\ea}{\end{eqnarray}}
\newcommand{\ban}{\begin{eqnarray*}}
\newcommand{\ean}{\end{eqnarray*}}

\newcommand{\ket}[1]{|#1\rangle}

\newcommand{\braket}[2]{\langle{#1}|{#2}\rangle}

\newcommand{\oo}{\otimes}

\begin{document}

\title{Knowledge by Direct Measurement versus Inference from Steering}

\author{Clive Aw}
\affiliation{Department of Physics, National University of Singapore, 2 Science Drive 3, Singapore 117542}
\author{Michele Dall'Arno}
\affiliation{Centre for Quantum Technologies, National University of Singapore, 3 Science Drive 2, Singapore 117543}
\author{Valerio Scarani}
\affiliation{Centre for Quantum Technologies, National University of Singapore, 3 Science Drive 2, Singapore 117543}
\affiliation{Department of Physics, National University of Singapore, 2 Science Drive 3, Singapore 117542}

\date{\today}

\begin{abstract}
If Alice and Bob start out with an entangled state $\ket{\Psi}_{AB}$, Bob may update his state to $\ket{\varphi}_B$ either by performing a suitable measurement himself, or by receiving the information that a measurement by Alice has steered that state. While Bob's update on his state is identical, his update on Alice's state differs: if Bob has performed the measurement, he has steered the state $\ket{\chi_{\leftarrow}(\varphi)}_A$ of Alice; if Alice has made the measurement, to steer $\ket{\varphi}_B$ on Bob she must have found a different state $\ket{\chi_{\rightarrow}(\varphi)}_A$. Based on this observation, a consequence of the well-known ``Hardy's ladder'', we show that information from direct measurement must trump inference from steering. The erroneous belief that both paths should lead to identical conclusions can be traced to the usual prejudice that measurements should reveal a pre-existing state of affairs. We also prove a technical result on Hardy's ladder: the minimum overlap between the steered and the steering state is $2\sqrt{p_{0}p_{n-1}}/(p_0+p_{n-1})$, where $p_0$ and $p_{n-1}$ are the smallest (non-zero) and the largest Schmidt coefficients of $\ket{\Psi}_{AB}$.
\end{abstract}
\maketitle

\section{Measurement and steering in a bipartite setting}

\subsection{A question}

A source produces the state $\ket{\Psi}_{AB}$ and sends one subsystem to Alice's lab, the other to Bob's lab. Bob leaves his assistants to take care of the task and goes to his office to perform some administrative duty. Shortly later, Bobby comes from the lab to inform Bob that, in some given rounds of the experiment, the state in their lab had been updated to $\ket{\varphi}_B$. Those particles, that have been kept in a quantum memory, are ready for use in subsequent tasks. Can Bob update also his knowledge of Alice's state in those same rounds? 

\subsection{Formalisation}

We write the initial state as
\ba
\ket{\Psi}_{AB}&=&\sum_k \sqrt{p_k}\,\ket{k}_A\otimes\ket{k}_B.
\ea and Bob's subsequent update as
\ba\ket{\varphi}_B&=&\sum_{k}\beta_k\ket{k}_B\,.\ea The probability $P_\beta=\sum_k p_k |\beta_k|^2$ of Bob finding this state is strictly positive, unless $\beta_k\neq 0$ only for indices $k$ such that $p_k=0$. 

We consider now two ways in which Bob's update may have come about. Suppose first that \textit{the measurement was done in Bob's lab}. It is a measurement in a basis that comprises $\ket{\varphi}$ and which, in that particular round, happened to yield that result. In this situation, Bob updates Alice's state to be the \textit{$\varphi-$steered state}
\ba
\ket{\chi_{\leftarrow}(\varphi)}_{A}&=&\sum_k \frac{\beta^*_k \sqrt{p_k}}{\sqrt{P_\beta}}\,\ket{k}_A\,.
\ea As second case, suppose that \textit{the measurement was done in Alice's lab} and Alice has informed Bobby that she has steered his state to $\ket{\varphi}_B$. This means that, in that round, Alice's measurement had yielded the \textit{$\varphi-$steering state}
\ba
\ket{\chi_{\rightarrow}(\varphi)}_{A}=\sum_k\alpha_k\ket{k}&\textrm{ with }& \frac{\alpha_k^*\sqrt{p_k}}{\sqrt{P_\alpha}}\,=\beta_k
\ea where $P_\alpha=\sum_k p_k|\alpha_k|^2$ is the probability of that outcome.

At this point, it is obvious for anyone familiar with Hardy's ladder \cite{hardy} that \textit{the two states $\ket{\chi_{\leftarrow}(\varphi)}_{A}$ and $\ket{\chi_{\rightarrow}(\varphi)}_{A}$ are generally different}. Indeed,
\ba\label{scalar}
\braket{\chi_{\leftarrow}(\varphi)}{\chi_{\rightarrow}(\varphi)}&=&\sqrt{\frac{P_\alpha}{P_\beta}}\,=\,\frac{\sum_k p_k |\alpha_k|^2 }{\left(\sum_k p_k^2 |\alpha_k|^2\right)^{1/2}}
\ea
is equal to 1 only in either of two cases: first, if all the $p_k$ are equal, i.e.~if $\ket{\Psi}_{AB}$ is maximally entangled; second, if $\alpha_{k}=\delta_{k,k'}$ for a given $k'$, which implies also $\beta_{k}=\delta_{k,k'}$, and means that the measurement (be it done by Alice or by Bob) is made in the Schmidt basis.

Also, the $\varphi$-steered and the $\varphi$-steering states are never orthogonal for a given $\ket{\Psi}_{AB}$. To the best of our knowledge, the minimum of the scalar product \eqref{scalar} was never reported for states of arbitrary dimensions. Using techniques from convex fractional optimisation (Appendix \ref{app}), we find
\ba
\min_{\varphi}\braket{\chi_{\leftarrow}(\varphi)}{\chi_{\rightarrow}(\varphi)}&=&\frac{2\sqrt{p_0p_{n-1}}}{p_0+p_{n-1}}
\ea where $p_0$ and $p_{n-1}$ are, respectively, the smallest and the largest non-zero Schmidt coefficients of $\ket{\Psi}_{AB}$. If $p_0$ and $p_{n-1}$ are not degenerate, the minimum is achieved for
\ba
\ket{\varphi}_B&=&\frac{1}{\sqrt{2}}\left(\ket{0}+e^{i\lambda}\ket{n-1}\right)
\ea
with arbitrary $\lambda\in\mathbb{R}$. This state steers
\ban
\ket{\chi_{\leftarrow}(\varphi)}_A&=&\sqrt{\frac{p_0}{p_0+p_{n-1}}}\ket{0}+e^{-i\lambda}\sqrt{\frac{p_{n-1}}{p_0+p_{n-1}}}\ket{n-1}\ean and is steered by
\ban
\ket{\chi_{\rightarrow}(\varphi)}_A&=&\sqrt{\frac{p_{n-1}}{p_0+p_{n-1}}}\ket{0}+e^{-i\lambda}\sqrt{\frac{p_{0}}{p_0+p_{n-1}}}\ket{n-1}\,.
\ean 

\section{Measurement trumps steering}

\subsection{Narrative of the previous observation}

Bob knows that, in the $n$ rounds under consideration, the state in his lab was $\ket{\varphi}_B$. Because of no-signaling, he can't know what Alice has done, if anything. At this stage, he can update his knowledge of Alice's state to the steered state $\ket{\chi_{\leftarrow}(\varphi)}_A$ in the sense that whatever Alice reports later won't be of zero probability given that state, the fact that \eqref{scalar} is never zero being a special instance of this. However, if later Alice informs Bob that the measurement was done in her lab, and that Bobby's knowledge came from Alice updating him \footnote{We note here that Bob can catch ``lazy Bobby'' by asking him to report the updated states in \textit{all} rounds, not just postselecting the rounds where a given state was found. Indeed, if Alice is steering and Bobby simply reporting, these states won't be orthogonal.}, then Bob must update Alice's state to the steering state $\ket{\chi_{\rightarrow}(\varphi)}_A$. This update will lead to more accurate predictions.

For a classical mind, there is something troubling in what we have written. After all, we have allowed Alice to perform one out of only two operations: either do nothing, or measure and tell Bob what is the state she steered on his side. That these two operations don't lead to the same state of knowledge means that \textit{Alice's measurement creates a state of affairs that could not have been known to Bob in advance} (and that's why he has to update again his knowledge). This is certainly counter-intuitive but not exactly new: it has been a tenet of quantum theory since the early days and was conclusively demonstrated by Bell's theorem.

\subsection{A variation, the same message}

The point may be reinforced by considering another set of rounds and possibly other measurements. We look at what happens when Bob's information comes from both sources: $\ket{\varphi}_B$ was found as a result of measurements in his lab, \textit{and} he is informed by Alice that she has steered his state to a possibly different state $\ket{\varphi'}_B$. Bob's should then \textit{update his knowledge to the outcomes of the two measurements}, namely $\ket{\varphi}_B$ for himself and $\ket{\chi_{\rightarrow}(\varphi')}_{A}$ for Alice. To see it, consider one possible chronology (thanks to no-signaling, timing does not matter). Bob measures first and gets $\ket{\varphi}_B$: he infers that he must have steered $\ket{\chi_{\leftarrow}(\varphi)}_A$ on Alice's side. Alice's message later informs him that she has made her own measurement and found $\ket{\chi_{\rightarrow}(\varphi')}_{A}$. The only states $\ket{\varphi'}_B$ for which the story is impossible on Alice's side are those such that $\braket{\chi_{\leftarrow}(\varphi)}{\chi_{\rightarrow}(\varphi')}= 0$. But $\braket{\chi_{\leftarrow}(\varphi)}{\chi_{\rightarrow}(\varphi')}=\frac{P'_{\alpha}}{P_{\beta}}\sum_{k}\beta_k^*\beta'_k$: the impossible states are those that are orthogonal to $\ket{\varphi}_B$, i.e.~those for which the story is impossible on Bob's side too.

\subsection{What happens if steering trumps measurement}

Failing to give direct measurement priority over inferences from steering may lead to absurd situations, the following \textit{\'echange de politesses} being an extreme one. Bob measures his system and finds $\ket{\psi_0}=\sum_k b_k\ket{k}$. He then informs Alice that he has steered her state to  $\ket{\psi_1}\propto\sum_k c_kb_k\ket{k}$, where for simplicity of notation we denote $c_k\equiv\sqrt{p_k}$, assume $a_k\in\mathbb{R}$, and omit normalisation. So far so good; but now Alice replies back as if \textit{she} had performed the measurement: if she has the state $\ket{\psi_1}$, then by steering Bob must have $\ket{\psi_2}\propto\sum_k c_k^2b_k\ket{k}$. Bob accepts Alice's inference on his system, then believes that he has done the measurement and informs Alice that her state must be $\ket{\psi_3}\propto\sum_k c_k^3b_k\ket{k}$; and so on. As soon as $b_k\neq\delta_{k,k'}$, the iteration's convergence is dominated by the largest Schmidt coefficient $p_{max}$ of $\ket{\Psi}_{AB}$: both the even (Bob's) and the odd (Alice's) sequence converge to $\ket{\psi_\infty}\propto \sum'_{k}b_k\ket{k}$ where the sum is on the indices $k$ such that $p_k=p_{max}$. The fact that Alice and Bob converge to an agreement may be desirable for peace but not for knowledge, since $\ket{\psi_\infty}$ has nothing to do with what they actually had in their labs (unless the initial state is maximally entangled, in which case all the $\ket{\psi_m}$ are equal).


\subsection{Relation to Frauchiger \& Renner's Thought Experiment}

The tension between updates from steering and updates from measurements may be detected in the argument put forward by Frauchiger and Renner (FR) \cite{FRPaper}, which is indeed a discussion of knowledge and certainty regarding measurements via the inference of various agents about each other's states. In this discussion, we mention some notations of that paper without explaining all of them.

In the FR thought experiment, the structure of the quantum state is \ban
\ket{\Psi}_{AB} &=& \sqrt{\frac{1}{3}}\Big(\ket{0}_A\ket{0}_B+\ket{0}_A\ket{1}_B+\ket{1}_A\ket{0}_B\Big)
\ean where $\ket{0}_A=\ket{\text{tails}}_R\oo\ket{\bar{t}}_{\bar{L}}$, $\ket{1}_A=\ket{\text{heads}}_R\oo\ket{\bar{h}}_{\bar{L}}$, $\ket{0}_B=\ket{\uparrow}_S\oo\displaystyle\ket{+{\sfrac{1}{2}}}_L$ and $\ket{1}_B=\ket{\downarrow}_S\oo\displaystyle\ket{-{\sfrac{1}{2}}}_L$.

The instances of steering and inference in question involve the time steps denoted $n:11 \to n:01$ and $n:01 \to n:31$ in Table 3 of the paper. Bob's measurement updated his state to $\ket{1}_B$: through steering, he infers that Alice's state is $\ket{0}_A$. When Alice is informed of this, by steering she would infer that Bob holds $\sqrt{\sfrac{1}{2}}(\ket{0}_B+\ket{1}_B)$. Suppose she communicates her inference to Bob, and Bob \textit{buys} this update rather than keeping the knowledge coming from his own measurement. Then, the paper's reasoning regarding the other two agents follow, namely: if $\bar w = \text{ok}$ then it is certain that $w = \text{fail}$; which translates to $P(\bar w = \bar{\text{ok}}, w = \text{ok}) = 0$, against the quantum prediction $P(\bar{w}= \bar{\text{ok}}, w = \text{ok})=\sfrac{1}{12}$. In other words, the FR argument exploits the minimal version of the \textit{\'echange de politesse} discussed before, where the dialogue stops at $\ket{\psi_2}$.

Also, as we argued above, the equivalence between updating from measurement and updating from steering could be assumed if measurement just reveals a pre-existing state of affairs. Thus, the similarity between the FR argument and Hardy's paradox in nonlocality may not be a mere mathematical incident, but likely stems from the same prejudice \cite{aaronson}.

\section{Conclusions}

We have shown that information from measurement must trump that from steering when updating an agent's knowledge on another agent's state. If this rule is not followed, paradoxical situation may appear. The fact that this rule is not trivial originates from the same prejudice that leads to formulating the local hidden variable assumption, namely, that measurements should just reveal a pre-existing state of affairs.

\section*{acknowledgments}

We acknowledge discussions with Koon Tong Goh. This research is supported by the National Research Fund and the Ministry of Education, Singapore, under the Research Centres of Excellence programme.


\begin{appendix}

\section{Minimisation of the scalar product \eqref{scalar}}\label{app}

In this appendix, for simplicity of notation we assume $\alpha_k\in\mathbb{R}$, without loss of generality. Also, the notation $x^*$ indicates the solution of an optimisation, rather than complex conjugation as in the main text.

\begin{thm}
  Let  $\mathbf{p}   \in  \mathbb{R}^n$  be   a  probability
  distribution such that $0 <  p_k \le p_{k+1}$ for any $k$.
  The following non-convex fractional optimization problem:
  \begin{align}
    \label{eq:fractional}
    \boldsymbol{\alpha}^*                                 :=
    \argmin_{\substack{\boldsymbol{\alpha}               \in
        \mathbb{R}^n\\\sum_k  \alpha_k^2 =  1}} \frac{\sum_k
      p_k \alpha_k^2} {\sqrt{\sum_k p_k^2 \alpha_k^2}},
  \end{align}
  is solved by any $\boldsymbol{\alpha}^*$ such that
  \begin{align*}
    \sum_{k    \in     \mathcal{K}_{\textrm{min}}}    \left(
    \alpha_k^*   \right)^2   &    =   \frac{p_{n-1}}{p_0   +
      p_{n-1}},\\  \sum_{k  \in  \mathcal{K}_{\textrm{max}}}
    \left(  \alpha_k^*   \right)^2  &  =   \frac{p_0}{p_0  +
      p_{n-1}},
  \end{align*}
  where $\mathcal{K}_{\textrm{min}}$           and
  $\mathcal{K}_{\textrm{max}}$ are the sets of indexes $k$'s
  such   that
  \begin{align*}
    p_k    <     p_j,    \qquad    &    \forall     k    \in
    \mathcal{K}_{\textrm{min}},   \;   \forall   j   \not\in
    \mathcal{K}_{\textrm{min}},\\  p_k   >  p_j,   \qquad  &
    \forall k  \in \mathcal{K}_{\textrm{max}}, \;  \forall j
    \not\in \mathcal{K}_{\textrm{max}},
  \end{align*}
  and $\alpha_k^* = 0$ for any other $k$.  In particular, if
  $|         \mathcal{K}_{\textrm{min}}|         =         |
  \mathcal{K}_{\textrm{max}}| = 1$, one has
  \begin{align*}
    \alpha_0^*    &    =    \pm\sqrt{\frac{p_{n-1}}{p_0    +
        p_{n-1}}},\\        \alpha_{n-1}^*        &        =
    \pm\sqrt{\frac{p_0}{p_0 + p_{n-1}}}.
  \end{align*}
  The figure of merit in Eq.~\eqref{eq:fractional} evaluates
  to
  \begin{align*}
    \frac{\sum_k    p_k    \left(   \alpha_k^*    \right)^2}
         {\sqrt{\sum_k p_k^2 \left(  \alpha_k^* \right)^2 }}
         = 2\frac{\sqrt{p_0 p_{n-1}}}{p_0 + p_{n-1}}.
  \end{align*}
\end{thm}

\begin{proof}
  First, notice  that for  any set $\mathcal{K}$  of indexes
  $k$'s  such   that  $p_k  =   p_j$  for  any  $k,   j  \in
  \mathcal{K}$,  by  direct  computation one  has  that  the
  figure of merit in Eq.~\eqref{eq:fractional} evaluates the
  same      for      any      $\boldsymbol{\alpha}$      and
  $\boldsymbol{\beta}$  such that  $\alpha_k =  \beta_k$ for
  any $k \not\in \mathcal{K}$, and
  \begin{align*}
    \sum_{k  \in \mathcal{K}}  \left(  \alpha_k \right)^2  =
    \sum_{k \in \mathcal{K}} \left( \beta_k \right)^2.
  \end{align*}
  Hence,  without  loss  of  generality, we  assume  $p_k  <
  p_{k+1}$ for any $k$.

  The remaining  of the proof  is lengthy and will  be split
  into the following three lemmas.
\end{proof}

\begin{lmm}
  Let  $\mathbf{p}   \in  \mathbb{R}^n$  be   a  probability
  distribution such  that $0 <  p_k < p_{k+1}$ for  any $k$.
  The optimization  problem in  Eq.~\eqref{eq:fractional} is
  equivalent to  the following optimization  problem, linear
  in $\mathbf{a}$:
  \begin{align}
    \label{eq:linear_in_a}
    \left(      \mathbf{a}^*,       s^*      \right)      :=
    \argmin_{\substack{\mathbf{a} \ge 0, s > 0\\\sum_k a_k =
        1\\ s^2 \sum_k p_k^2 a_k = 1}} s \sum_k p_k a_k,
  \end{align}
  where $\alpha_k = \pm \sqrt{a_k}$ for any $k$.
\end{lmm}

\begin{proof}
  Equation~\eqref{eq:fractional}    is   an    instance   of
  fractional programming~\cite{Sch74, Sch83}.  The numerator
  and  the denominator  of the  figure of  merit are  convex
  functions,  however the  constraint is  not a  convex set.
  Hence,  Eq.~\eqref{eq:fractional} is  not  an instance  of
  \textit{convex}    fractional    programming.     However,
  Eq.~\eqref{eq:fractional}  can  be   recast  as  a  convex
  fractional   programming  by   means   of  the   following
  substitution.  By setting $a_k  := \alpha_k^2$ for any $k$
  one has  that $\alpha_k^* = \pm\sqrt{a_k^*}$  for any $k$,
  where
  \begin{align}
    \label{eq:convex_fractional}
    \mathbf{a}^*   :=    \argmin_{\substack{\mathbf{a}   \ge
        0\\\sum_k   a_k   =   1}}  \frac{\sum_k   p_k   a_k}
           {\sqrt{\sum_k p_k^2 a_k}}.
  \end{align}
  Equation~\eqref{eq:convex_fractional}  is now  an instance
  of  convex fractional  programming.   A  subset of  convex
  fractional programming for which special results hold (see
  Case  1 at  the  end  of page  3  of Ref.~\cite{Sch74}  or
  Proposition 7  of Ref.~\cite{Sch83}) is the  case in which
  the denominator of the figure of merit is affine, which is
  not the case in Eq.~\eqref{eq:convex_fractional}. However,
  this  can be  amended  by  another simple  transformation.
  Since the figure of merit is non-negative on the domain of
  optimization, taking its square one has
  \begin{align}
    \label{eq:convex_fractional_with_linear_denominator}
    \mathbf{a}^*   :=    \argmin_{\substack{\mathbf{a}   \ge
        0\\\sum_k  a_k =  1}}  \frac{\left(  \sum_k p_k  a_k
      \right)^2} {\sum_k p_k^2 a_k}.
  \end{align}
  Notice  that  the numerator  and  the  denominator of  the
  figure  of  merit are  a  convex  and a  linear  function,
  respectively, and  that the optimization is  over a convex
  set.                                                Hence,
  Eq.~\eqref{eq:convex_fractional_with_linear_denominator}
  is  an  instance  of convex  fractional  programming  with
  affine denominator.  It  was shown (see Case 1  at the end
  of  page  3  of  Ref.~\cite{Sch74}  or  Proposition  7  of
  Ref.~\cite{Sch83}) that by setting $a_k = b_k / t$ for any
  $k$ one has $a_k^* = b_k^* / t^*$ for any $k$, where
  \begin{align}
    \label{eq:convex}
    \left(      \mathbf{b}^*,       t^*      \right)      :=
    \argmin_{\substack{\mathbf{b}  \ge 0,  \; t  > 0\\\sum_k
        \frac{b_k}t = 1\\ t \sum_k p_k^2 \frac{b_k}t = 1}} t
    \left( \sum_k p_k \frac{b_k}t \right)^2.
  \end{align}
  Notice    that    Eq.~\eqref{eq:convex}   represents    an
  optimization problem convex in variable $\mathbf{b}$.  Yet
  another        simple        transformation        recasts
  Eq.~\eqref{eq:convex} as an optimization problem linear in
  $\mathbf{b}$.  Since  the figure of merit  is non-negative
  on the domain of optimization,  taking its square root one
  has
  \begin{align}
    \label{eq:linear_in_b}
    \left(      \mathbf{b}^*,       t^*      \right)      :=
    \argmin_{\substack{\mathbf{b}  \ge 0,  \; t  > 0\\\sum_k
        \frac{b_k}t =  1\\ t \sum_k p_k^2  \frac{b_k}t = 1}}
    \sqrt{t} \sum_k p_k \frac{b_k}t.
  \end{align}
  Finally,  to  recast   Eq.~\eqref{eq:linear_in_b}  as  the
  optimization problem in Eq.~\eqref{eq:linear_in_a}, we set
  $s :=  \sqrt{t}$ and substite  back $a_k = b_k/t$  for any
  $k$.
\end{proof}

\begin{lmm}
  The optimization problem  in Eq.~\eqref{eq:linear_in_a} is
  equivalent to the following scalar optimization problem:
  \begin{align}
    \label{eq:scalar0}
    \left(    k_0^*,   k_1^*,    s^*,    a^*   \right)    :=
    \argmin_{\substack{k_0,  \; k_1,  \;  s >  0,  \; a  \ge
        0\\s^2\left(     p_{k_0}^2     a     +     p_{k_1}^2
        \left(1-a\right) \right) = 1}}  s \left( p_{k_0} a +
    p_{k_1} (1-a) \right).
  \end{align}
  where $a^*_{k_0^*}  = a^*$ and $a^*_{k_1^*}  = 1-a^*$, and
  $a_k^* = 0$ for any $k \not= k_0^*, k_1^*$.
\end{lmm}
  
\begin{proof}
  Since $s$ and  $p_k$'s appear in terms of  the same degree
  in  the  figure  of  merit   and  in  the  constraints  in
  Eq.~\eqref{eq:linear_in_a},  by setting  $q_k :=  s^* p_k$
  for any $k$, where $s^* = \sqrt{t^*}$, one has
  \begin{align}
    \label{eq:linear}
    \mathbf{a}^*   :=    \argmin_{\substack{\mathbf{a}   \ge
        0\\\sum_k a_k  = 1\\ \sum_k  q_k^2 a_k =  1}} \sum_k
    q_k a_k,
  \end{align}
  which is a linear optimization problem whose constraint is
  a polytope  given by  the intersection of  the probability
  simplex with the hyperplane $\sum_k q_k^2 a_k = 1$.

  The (possibly  local) extrema  of the  figure of  merit in
  Eq.~\eqref{eq:linear} under its  constraints are either in
  the bulk of  such a polytope or on its  boundary, that is,
  when at least one of the elements of $\mathbf{a}$ is zero.
  In the  former case, extrema  can be found by  looking for
  extrema when the inequality  constraint $\mathbf{a} \ge 0$
  is  relaxed,  and  selecting  those that  lie  inside  the
  polytope.  In  the latter  case, extrema  can be  found by
  setting  one   element  of   $\mathbf{a}$  to   zero,  and
  proceeding   as   in   the  previous   case.    Proceeding
  recursively,  one  needs  to  look for  extrema  when  any
  possible subset of the elements of $\mathbf{a}$ are set to
  zero.  In the following, we  use the technique of Lagrange
  multipliers to show  that any such an extrema  has at most
  two non-null elements.

  By introducing  Lagrange multipliers $\lambda$  and $\mu$,
  for  any  set  $\mathcal{K}$  subset of  the  set  of  all
  possible indexes $k$'s, that is $\mathcal{K} \subseteq [0,
    n-1]$, one can write the following auxiliary function
  \begin{align*}
    \mathcal{L}  \left( \mathcal{K}  \right) :=  \sum_{k \in
      \mathcal{K}} \left( q_k +  \lambda + \mu q_k^2 \right)
    a_k.
  \end{align*}
  Hence,  a necessary  condition  for $\mathbf{a}$  to be  a
  (possibly  local)  extrema  of  the  figure  of  merit  in
  Eq.~\eqref{eq:linear} over its constraint is that
  \begin{align}
    \label{eq:lagrange}
    \frac{\partial}{\partial a_k} \mathcal{L}(\mathcal{K}) =
    0, \qquad \forall k,
  \end{align}
  for at least one set $\mathcal{K} \subseteq [0, n-1]$.  By
  explicit computation one has
  \begin{align*}
    \frac{\partial}{\partial    a_k}   \mathcal{L}    \left(
    \mathcal{K} \right) = \begin{cases} \left( q_k + \lambda
      +  \mu  q_k^2 \right),  &  \quad  \textrm{ if  $k  \in
        \mathcal{K}$,}\\     0,     &     \quad     \textrm{
        otherwise.}\end{cases}
  \end{align*}
  Since   all  $q_k$'s   are   different,   the  system   in
  Eq.~\eqref{eq:lagrange} contains  $|\mathcal{K}|$ linearly
  independent  equations in  variables $\lambda$  and $\mu$,
  hence  such  a system  admits  solutions  if and  only  if
  $|\mathcal{K}| \le 2$. Hence the statement follows.
\end{proof}

\begin{lmm}
  The  optimization  problem  in  Eq.~\eqref{eq:scalar0}  is
  solved by $k_0^* = 0$, $k_1^* = n - 1$, and
  \begin{align*}
    a^* = \frac{p_{k_1^*}}{p_{k_0^*} + p_{k_1^*}}.
  \end{align*}
\end{lmm}

\begin{proof}
  Let us first solve the problem in $a$ for any given $k_0$,
  $k_1$, and $s$. Form the constraint, by direct computation
  one has
  \begin{align}
    \label{eq:astar}
    a^* =  \frac{1 -  s^2 p_{k_1}^2}{s^2 \left(  p_{k_0}^2 -
      p_{k_1}^2 \right)}.
  \end{align}

  Let us  now solve the problem  in $s$ for any  given $k_0$
  and  $k_1$.   By  substituting  Eq.~\eqref{eq:astar}  into
  Eq.~\eqref{eq:scalar0} one immediately has
  \begin{align}
    \label{eq:scalar1}
    \left(  k_0^*, k_1^*,  s^* \right)  := \argmin_{k_0,  \;
      k_1, \; s >  0} \frac1{p_{k_0}+p_{k_1}} \left( \frac1s
    + s p_{k_0} p_{k_1} \right).
  \end{align}
  By explicit computation,  the figure of merit  is a convex
  function in $s$ and is thus minimized in $s$ by taking the
  zero  of   its  first   derivative.  Hence,   by  explicit
  computation one has
  \begin{align}
    \label{eq:sstar}
    s^* = \frac1{\sqrt{p_{k_0} p_{k_1}}}.
  \end{align}

  Let us finally solve the problem in $k_0$ and $k_1$.  Upon
  replacing Eq.~\eqref{eq:sstar} into Eq.~\eqref{eq:scalar1}
  one has
  \begin{align}
    \label{eq:scalar2}
    \left( k_0^*,  k_1^* \right) := \argmin_{k_0,  \; k_1} 2
    \frac{\sqrt{p_{k_0} p_{k_1}}}{p_{k_0} + p_{k_1}}.
  \end{align}
  Without loss of  generality, let us take $k_0  < k_1$, and
  hence  $p_{k_0} <  p_{k_1}$.   Since $k_0$  and $k_1$  are
  discrete variables, one cannot directly apply optimization
  techniques  based on  differential  methods.  However,  it
  follows by  direct computation  that, upon defining  $r :=
  p_{k_0}   /    p_{k_1}$,   the   figure   of    merit   in
  Eq.~\eqref{eq:scalar2} can be written as
  \begin{align}
    \label{eq:discrete}
    \frac{\sqrt{p_{k_0}   p_{k_1}}}{p_{k_0}  +   p_{k_1}}  =
    \frac{\sqrt{r}}{r + 1}.
  \end{align}
  Hence,  the  figure  of  merit  in  Eq.~\eqref{eq:scalar2}
  depends  on $p_{k_0}$  and  $p_{k_1}$  only through  their
  ratio $r$.

  We can now apply differential methods to variable $r$.  By
  explicit computation, one has that the first derivative of
  Eq.~\eqref{eq:discrete} in $r$ is positive in the range $0
  \le  r <  1$  that  we are  considering  since $p_{k_0}  <
  p_{k_1}$.     Hence,    the    figure    of    merit    in
  Eq.~\eqref{eq:scalar2} is monotonically increasing in $r$,
  and is thus minimized by the minimal $r$. By definition of
  $r$, this is achieved by $k_0^* = 0$ and $k_1^* = n-1$.
\end{proof}
\end{appendix}


\begin{thebibliography}{}

\bibitem{hardy} L. Hardy, \textit{Nonlocality for two particles without inequalities for almost all entangled states}, Phys. Rev. Lett. \textbf{71}, 1665 (1993)

\bibitem{FRPaper}
D. Frauchiger, R. Renner,
\textit{Quantum theory cannot consistently describe the use of itself}, Nature Communications \textbf{9}, 3711 (2018). 

\bibitem{aaronson} S. Aaronson, \textit{It's hard to think when someone Hadamards your brain}, 25 Sept 2018 \href{https://www.scottaaronson.com/blog/?p=3975}{(https://www.scottaaronson.com/blog/?p=3975)}

\bibitem{Sch74} S.   Schaible, \textit{Parameter-free Convex
  Equivalent  and Dual  Programs  of Fractional  Programming
  Problems},    Zeitschrift   f\"ur    Operations   Research
  \textbf{18}, 187 (1974).
\bibitem{Sch83}     S.      Schaible,     \textit{Fractional
  Programming},   Zeitschrift   f\"ur  Operations   Research
  \textbf{27}, 39 (1983).
\end{thebibliography}
\end{document}